\documentclass[11pt]{article}
\usepackage{graphicx} 

\usepackage{fullpage}
\usepackage{amsmath}
\usepackage{amsthm}
\usepackage{amssymb}
\usepackage{graphicx}

\usepackage{xcolor}
\usepackage{hyperref}
\usepackage[capitalize,noabbrev]{cleveref}
\hypersetup{
    colorlinks,
    linkcolor={red!50!black},
    citecolor={blue!50!black},
    urlcolor={blue!80!black}
}

\usepackage[sortcites,maxbibnames=99]{biblatex}
\newtheorem{theorem}{Theorem}[section]
\newtheorem{corollary}[theorem]{Corollary}
\newtheorem{lemma}[theorem]{Lemma}

\title{Bounding a Polygon by a Minimum Number of Vertices}
\author{Mikkel Abrahamsen \and Jack Stade \and Shuyi Yan \and Hanwen Zhang}

\newcommand\blfootnote[1]{%
  \begingroup
  \renewcommand\thefootnote{}\footnote{#1}%
  \addtocounter{footnote}{-1}%
  \endgroup
}
\date{}

\begin{document}

\maketitle
\blfootnote{MA, JS, SY and HZ: University of Copenhagen, \texttt{\{miab,jast,yash,hazh\}@di.ku.dk}.
MA, JS and HZ are supported by Starting Grant 1054-00032B from the Independent Research Fund Denmark under the Sapere Aude research career programme.
MA, JS, YS and HZ are part of Basic Algorithms Research Copenhagen (BARC), supported by the VILLUM Foundation grant 54451.}

\begin{abstract}
Suppose that a polygon $P$ is given as an array containing the vertices in counterclockwise order.
We analyze how many vertices (including the index of each of these vertices) we need to know before we can bound $P$, i.e., report a bounded region $R$ in the plane such that $P\subset R$.
We show that there exists polygons where $4\log_2 n+O(1)$ vertices are enough, while $\log_3n-o(\log n)$ must always be known.
We thus answer the question up to a constant factor.
This can be seen as an analysis of the shortest possible certificate or the best-case running time of any algorithm solving a variety of problems involving polygons, where a bound must be known in order to answer correctly.
This includes various packing problems such as deciding whether a polygon can be contained inside another polygon.
\end{abstract}

\section{Introduction}

This research was motivated by Marc van Kreveld's slides \cite{slides} for a course in computational geometry at Utrecht University.
The slides contain the following questions: \emph{``Suppose that a simple polygon with $n$ vertices is given;
the vertices are given in counterclockwise order along the
boundary. Give an efficient algorithm to determine all edges that
are intersected by a given line.
How efficient is your algorithm? Why is your algorithm efficient?''}
A naive $O(n)$-time algorithm that checks whether each edge intersects the line is worst-case optimal, which is perhaps the intended answer to the last question.
Indeed, in an example as in \Cref{fig:worstcase}, we need to check both endpoints of nearly all $n$ edges in order to solve the problem.

\begin{figure}[ht]
\centering   \includegraphics[page=2]{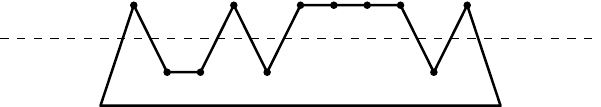}
\caption{Each of the marked vertices could be above or below the line, so to find the intersecting edges, we need to consider all of these vertices.}
\label{fig:worstcase}
\end{figure}

This made us wonder whether $O(n)$ is also the \emph{best-case} running time.
In particular, what is the minimum number of vertices considered by any correct algorithm on any input, consisting of a polygon $P$ and a line?
The algorithm may query the position of any vertex 
$P[i]$, and crucially, knowing some of the vertices may restrict the placement of the others, since the polygon must remain simple---its edges cannot cross. 
Even then, we were initially certain that it would be impossible for any algorithm to answer correctly after inspecting only $o(n)$ vertices.

To our own surprise we managed to construct polygons where you only need to know $4\log_2 n+O(1)$ vertices in order to report the number of edges intersecting a line.
We subsequently managed to show that any algorithm must consider at least $\log_3 n - o(\log n)$ vertices in order to answer correctly.
We have thus found the minimum up to constant factors.
The problem we are actually considering is how many vertices we need to know before we can \emph{bound} $P$, i.e., report a bounded region $R$ in the plane so that $P\subset R$; hence the title of the paper.
If the given line $\ell$ is disjoint from $R$, we can then report that there are no intersecting segments. 

There are several other problems involving polygons that are closely linked to the bounding problem, where our result implies a $\Theta(\log n)$ bound on the shortest certificate or, equivalently, the best-case running time among all algorithms.
This for instance includes most packing problems:
Suppose that we want to decide if a given object (or set of objects) can be placed in a polygon $P$.
If we know a bound $R$ on $P$ and the objects don't fit in $R$, we can immediately report that they don't fit in $P$ either.
On the other hand, if we have no bound on $P$, then the objects may fit, so we can't answer before a bound is (implicitly) known. 
Thus, the $\Theta(\log n)$ bound again describes the length of the shortest possible certificate.

To mention another example where a bound is used positively, suppose that we want to decide if a polygon $P$ fits in a given square. 
If we know a bound $R$ on $P$ and $R$ fits in the square, we can answer in the affirmative, but if we can't even derive any bound, it is impossible to answer correctly (provided that the vertices seen so far do not already reveal that $P$ is too large).
Hence, the $\Theta(\log n)$ bound again captures the best case.
For some recent references to problems of this type, see~\cite{DBLP:conf/soda/KunnemannN22,DBLP:journals/theoretics/AbrahamsenMS24,DBLP:conf/focs/AbrahamsenS24,DBLP:conf/soda/AbrahamsenR25,DBLP:conf/esa/ArkinD0GMPT20,DBLP:conf/approx/FeketeKKMS11}.

A key concept in our work is \emph{link distance}, which has been studied since the early days of computational geometry; see the survey in \cite[p.~822]{DBLP:reference/cg/Mitchell04} and \cite{DBLP:journals/jocg/KostitsynaLPS17,DBLP:journals/comgeo/MitchellPS14} for some of the recent references.
The link distance between two given points is the minimum number of edges needed to connect them with a polygonal path that avoids some obstacles.
In our construction, we use link distance to argue that the polygon must stay in a bounded region.
Knowing some vertices of $P$, it may be possible to conclude that in order to connect two points $p$ and $q$ without crossing other edges, we need at least $k$ edges.
If we are then told that vertex $i$ is the point $p$ and vertex $i+k$ is $q$, then the whole chain from vertex $i$ to $i+k$ must follow a minimum link path from $p$ to $q$, and this greatly restricts all of the vertices $i+1,\ldots,i+k-1$.

In \Cref{sec:construction}, we will describe a polygon $P$ with $n$ vertices where $4\log_2 n+O(1)$ vertices restrict the entire polygon to a bounded region.
In \Cref{sec:lowerbound}, we show that no polygon can be bounded by just $\log_3 n - \log_3\log_3 n - 1$ vertices.
We have tried to keep both upper and lower bound as simple as possible and thus not made a serious attempt to find the correct leading constant, but this remains an interesting question.

\subsection{Preliminaries}

A \emph{polygonal path} is a simple open curve in the plane consisting of a finite number of line segments.
A \emph{(simple) polygon} is the closed region in the plane enclosed by a simple closed curve consisting of a finite number of line segments.
These segments are the \emph{edges} of the polygon, and their endpoints are the \emph{vertices}.

A polygon $P$ with $n$ vertices is represented by an array of points, so that the vertices are $P[0],\ldots,P[n-1]$ in counterclockwise order along the boundary of $P$.
For indices $a,b$ with $a< b$, the \emph{chain} $P[a:b]$ is the polygonal path from $P[a]$ to $P[b]$ along the boundary in counterclockwise order.
If the segment $P[a]P[b]$ does not cross the chain $P[a:b]$, then $P[a:b]\cup P[a]P[b]$ encloses a simple polygon that we denote $\bar P[a:b]$.

Given a set of edges $E$ and two points $p$ and $q$, the \emph{link distance} between $p$ and $q$ with respect to $E$ is the minimum number of edges on a polygonal path $\pi$ from $p$ to $q$ so that $\pi$ is disjoint from $E$, except possibly at $p$ and $q$.
A \emph{minimum link path} is such a path $\pi$.

A \emph{partial description} $P$ of a polygon is an array of points where some entries have been specified while the rest have not.
A \emph{realization} $Q$ of $P$ is a polygon that agrees on the specified vertices.
We say that $P$ is \emph{realizable} if there exists a realization of $P$.
We can likewise talk about realizability of a chain $P[a:b]$, which means that the unspecified vertices on the chain can be chosen so that the chain does not intersect itself and also avoids another set of edges which will be clear from the context.

A \emph{bound} on a polygon $P$ is a bounded region $R$ in the plane so that $P\subset R$.
Likewise, a bound on a partial description $P$ is a bounded region $R$ that contains all realizations of $P$.
Note that if $P$ cannot be realized, then any bounded region $R$ is trivially a bound on $P$.

\section{Construction}\label{sec:construction}

In this section we give a realizable partial description $P$ of a polygon with more than $2^j$ vertices, for any integer $j$, where we can give a bound on $P$ after specifying only $4j+O(1)$ vertices.
We first need some notation.
For an index $i$, define $\Delta_i$ to be the triangle $P[i-1]P[i]P[i+1]$ and define $\square_i$ to be the quadrilateral $P[0]P[1]P[i-1]P[i]$.
We will only use the notation $\square_i$ when the vertices form a convex quadrilateral.

In the following two lemmas, we consider the following setting.
Let $a$ and $b$ be indices with $a>2$ and $b>a+2$ so that in any realization of $P[0:b]$, we have the following properties:
\begin{itemize}
\item $\square_b$ is convex,
\item the polygon $\bar P[0:b]$ is simple and contains $\Delta_a$ and $\square_b$,
\item the link distance with respect to $P[0:b]$ from any point in $\square_b$ to any point in $\Delta_a$ is at least~$k$.
\end{itemize}
The setting is illustrated in \Cref{fig:doubling}.
Our first lemma explains how we can double the link distance by specifying only a constant number of new vertices.

\begin{figure}
\centering   \includegraphics[page=1]{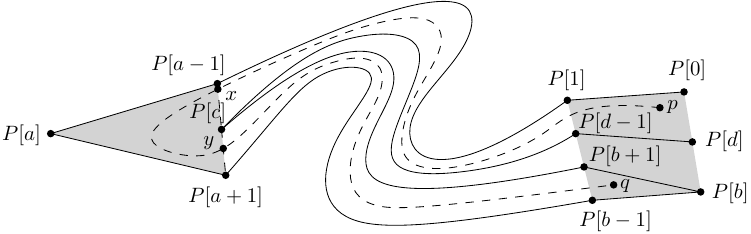}
\caption{Illustration of \Cref{lemma:doubling}.
The dashed curve is $\pi$ from the proof.}
\label{fig:doubling}
\end{figure}

\begin{lemma}\label{lemma:doubling}
In the situation described above, let $P[b+1]\in P[1]P[b-1]$,  $P[c]\in P[a-1]P[a+1]$, $P[d-1]\in P[1]P[b+1]$, and $P[d]\in P[0]P[b]$, for indices $c,d$ with $d>c>b+1$.
Suppose that in any realization of $P[0:d]$, $P[b+1:d-1]$ is contained in $\bar P [1:b-1]$ and is disjoint from the interiors of both $\square_b$ and $\Delta_a$.
Then $\square_d$ is convex and $\bar P[0:d]$ is simple and contains $\Delta_b$ and $\square_d$, and the link distance with respect to $P[0:d]$ from any point in $\square_d$ to any point in $\Delta_b$ is at least $2k$.
\end{lemma}

\begin{proof}
Note that $\square_d$ is convex because $\square_b$ is convex, and $\bar P[0:d]$ is simple because $P[0:b]$ is simple, $P[d]\in P[0]P[b]$ and $P[b+1:d-1]$ is contained in $\bar P[1:b-1]$.
For the same reasons, $\bar P[0:d]$ contains $\square_d$, and since $P[b+1:d-1]$ is disjoint from the interior of $\Delta_a$, we also have $\Delta_a\subset \bar P[0:d]$.

To see that the link distance with respect to $P[0:d]$ from any point in $\square_d$ to any point in $\Delta_b$ is at least $2k$, consider a polygonal path $\pi$ from $p\in\square_d$ to $q\in\Delta_b$ that avoids $P[0:d]$.
Since $P[a-1:c] \cup P[a-1]P[c]$ separates $p$ from $q$, $\pi$ must cross a point $x$ along $P[a-1]P[c]$. 
Similarly, since $P[a+1:c] \cup P[a+1]P[c]$ separates $p$ from $q$, $\pi$ must cross a point $y$ along $P[a+1]P[c]$. 
According to our hypothesis, both $\pi[p,x]$ and $\pi[y,q]$ contains at least $k$ segments. 
Note that the segment $xy$ is blocked by $P[c]$, so there must be a vertex of $\pi$ between $x$ and $y$ along $\pi[x,y]$. 
Therefore, $\pi = \pi[p,x] \cup \pi[x,y] \cup \pi[y,q]$ contains at least $2k$ segments. 
This proves the lemma.
\end{proof}

We will choose $P[b+1]$ and $P[c]$ as a pair of points from segments $P[1]P[b-1]$ and $P[a-1]P[a+1]$, respectively, that minimizes the link distance (also minimizing over all realizations of $P[0:b]$).
If $k$ is this link distance, we define $c=b+1+k$.
As we will see, this definition of $P[b+1]$ and $P[c]$ will ensure that in any realization of $P[0:d]$, the chain $P[b+1:d-1]$ enters neither $\square_b$ nor $\Delta_a$.

Note that the realization space of minimum link paths is open: all vertices can be slightly moved in any direction without colliding with the chain $P[0:b]$.
Hence, we can choose $d=c+k+1$ and define $P[d-1]$ as a point in $P[1]P[b+1]$ with the same link distance from $P[c]$ as $P[b+1]$.
We finally define $P[d]=(P[0]+P[b])/2$.
It is not immediately clear whether the whole chain $P[b:d]$ can be realized, because the chains to and from $P[c]$ could cross each other.
However, the following lemma states that it is realizable when choosing $P[d-1]$ correctly and, furthermore, that $P[b:d]$ stays in the polygon $\bar P[0:b]$.

\begin{lemma}\label{lem:realizablecontainment}
Choosing $P[d-1]$ on $P[1]P[b+1]$ sufficiently close to $P[b+1]$, the chain $P[b:d]$ can be realized.
In all realizations, the chain $P[b+1:d-1]$ is contained in $\bar P[1:b-1]$ and is disjoint from the interior of $\Delta_a$, and we have $\bar P[0:d]\subset \bar P[0:b]$.
\end{lemma}

\begin{proof}
We choose $P[d-1]$ on $P[1]P[b+1]$ so close to $P[b+1]$ that the chain $P[d-1]P[b+2:c]$ also avoids $P[0:b]$.
We can then choose the vertices of $P[c+1:d-2]$ on the angular bisectors the vertices $P[b+2:c-1]$; see \Cref{fig:bisector}.
This shows that the whole chain $P[b:d]$ is realizable.

\begin{figure}
\centering   \includegraphics[page=3]{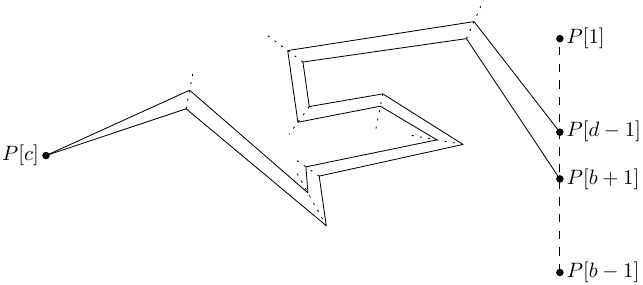}
\caption{Choosing $P[d-1]$ sufficiently close to $P[b+1]$ and the vertices $P[c+1:d-2]$ on the bisectors of $P[b+2:c-1]$, we can ensure that the chain $P[b:d]$ is realizable.}
\label{fig:bisector}
\end{figure}

If the chain $P[b+1:d-1]$ is not in the polygon $\bar P[1:b-1]$, then $P[b+1:d-1]$ leaves $\bar P[1:b-1]$ through the segment $P[1]P[b-1]$ and enters again through a point $y$ in order to get to $P[c]$.
Thus, the link distance from $y$ to $P[c]$ is smaller than that from $P[b+1]$, contradicting the choice of $P[b+1]$.
Likewise, $P[b+1:d-1]$ is disjoint from the interior of $\Delta_a$, since otherwise the chain crosses the segment $P[a-1]P[a+1]$ at a point with smaller link distance to $P[b+1]$ than $P[c]$, which contradicts the choice of $P[c]$.
It thus follows that $P[b:d]$ is in $\bar P[0:b]$, so we conclude that $\bar P[0:d]\subset\bar P[0:b]$.
\end{proof}

We can now give a recursive realizable partial description of a polygon with more than $2^j$ vertices, of which only $4j+O(1)$ are specified, where a bound also can be given.
Our starting point is a convex heptagon $\bar P[0:6]$.
We then define $a_0=3$ and $b_0=6$.
We now repeat the following pattern for $i=0,\ldots,j$.
We construct points with indices $b_i+1,c_i,d_i-1,d_i$ from the chain $P[0:b_i]$ as described above.
We then define $a_{i+1}=b_i$ and $b_{i+1}=d_i$.

\begin{lemma}
Suppose that we have specified vertices of a polygon $P$ as described above.
The polygon $P$ can be realized, and every realization is contained in the heptagon $\bar P[0:6]$.
The polygon has $d_j+1> 2^j$ vertices, although only $4j+O(1)$ have been specified.
\end{lemma}

\begin{proof}
By \Cref{lemma:doubling}, each chain $P[b_{i+1}+1:d_{i+1}-1]$ consists of at least twice as many edges as the predecessor $P[b_i+1:d_i-1]$.
It hence follows that $d_j$ is at least $1+2+\ldots+2^j>2^j$.
Each time we increase $j$, we define $4$ new vertices, so in total there are $4j+O(1)$ explicitly defined vertices.

By \Cref{lem:realizablecontainment}, each of the chains $P[b_i,d_i]$ can be realized, so the whole polygon can be realized, and we have
\[
\bar P[0,d_j]\subset \bar P[0,d_{j-1}] \subset \cdots \subset \bar P[0,d_0]\subset\bar P[0,6].
\]
\end{proof}

This immediately implies:

\begin{corollary}
For arbitrarily large values of $n$, there exists a polygon with $n$ vertices so that after specifying $4\log_2 n+O(1)$ vertices, a bound on the polygon can be given.
\end{corollary}

\section{Lower Bound}\label{sec:lowerbound}

In this section, we show that any realizable partial description $P$ of a polygon with $n$ vertices must specify at least $\log_3 n - \log_3\log_3 n$ vertices to be bounded. If there are less specified points, we can construct a realization containing an arbitrary point on its boundary, so the partial description is unbounded.

The following simple lemma will be useful in our construction.

\begin{lemma}
    \label{lem:link-distance}
    Given a set of polygonal obstacles in the plane with $m$ vertices in total, for any pair $(x, y)$ of reachable points in the plane, 
    the link distance from $x$ to $y$ is at most $m+1$. 
\end{lemma}

\begin{proof}
    Consider the shortest path in the Euclidean distance from $x$ to $y$ that avoids the interior of the obstacles. 
    A folklore result says that this path is polygonal, and every internal vertex in the path is a vertex of some obstacle. 
    Therefore, this path consists of at most $m+1$ segments. 
    With a proper perturbation on all the vertices, we can make the path lie in the free space. 
\end{proof}

To construct the realization, we look at the (currently unknown) chains between specified points. Since we know the indices of the specified points, we know the lengths of these chains. Intuitively, if there are some long chains that are much longer than (even the sum of) the other short chains, then, no matter how the short chains look like, there will be enough freedom to construct the long chains that can reach anywhere in the plane.

Consider a realizable partial description $P$ of a polygon with $n$ vertices, which specifies $k$ points.
Let $0 \le i_1 < i_2 < \dots < i_k < n$ be the indices of the specified points.
Let $a_j = i_{j + 1} - i_{j}$ for $j \in [k-1]$ and $a_k = i_1 + n - i_k$, i.e. $a_j$ is the length of the chain from $P[i_j]$ to $P[i_{j+1}]$. Let $a_{(j)}$ be the $j$-th smallest value among $a_1, \dots, a_k$.
Let $x_j = k + \sum_{i = 1}^j (a_{(i)} - 1)$ for $j = 0, 1, \dots, k$. 

\begin{lemma}
    \label{lem:bound}
    If there exists $i \in \{0, 1, \dots, k-1\}$ such that for every $j \in \{i+1, \dots, k-1\}$, $a_{(j)} > 2^{j-i-1} x_{i}$, and $a_{(k)} > 3 \cdot 2^{k-i-1}x_i$, then for any point $v$, there exists a realization $Q$ with $v$ on its boundary. 
\end{lemma}

\begin{proof}
    We will construct the chains in the non-decreasing order of their lengths. 
    For the chains corresponding to $a_{(1)}, \dots, a_{(i)}$, we follow an arbitrary realization of $P$. 
    Starting from $a_{(i+1)}$, we iteratively use the minimum link path, considering all existing segments as obstacles. 

    Now we start to construct the chains corresponding to $a_{(i+1)}, \dots, a_{(k)}$ one by one, each with a possibly shorter chain than the desired length. 
    Before constructing the chain corresponding to $a_{(i+1)}$, the total number of vertices in the obstacles is $x_{i}$. 
    By \cref{lem:link-distance}, we can always construct the chain corresponding to $a_{(i+1)}$ with at most $x_{i}$ new vertices. 
    So in each step, the complexity of obstacles at most doubles. By induction, since $a_{(j)} > 2^{j-i-1}x_{i}$ for all $j \in \{i+1,\dots,k-1\}$, we can successfully construct all the chains $a_{(i+1)}, \dots, a_{(k-1)}$, and the total amount of vertices we used for $a_{(1)}, \dots, a_{(k-1)}$ is at most
    $$x_i+\sum_{j=i+1}^{k-1}2^{j-i-1}x_i=2^{k-i-1} x_i.$$

    The next step is to enclose the longest gap corresponding to $a_{(k)}$. 
    At this moment, the obstacles form a simple path. 
    Let $x$ and $y$ be the two ends of the obstacles. 
    Recall that the number of vertices in this path is at most $2^{k-i-1} x_{i}$.
    If $v$ is already on some obstacles, we only need to enclose the polygon with a sufficiently close path following the obstacles as in the proof of \cref{lem:realizablecontainment}, which only need $2^{k-i-1} x_{i} < a_{(k)}$ segments. 
    Otherwise, we will include $v$ as a vertex in the longest chain. 
    We can first connect $y$ to $v$ with a path of at most $2^{k-i-1} x_{i}$ new vertices, by \cref{lem:link-distance}. 
    After that, we will enclose the polygon with a sufficiently close path from $v$ to $x$ following the existing path and obstacles, by \cref{lem:realizablecontainment}, which requires $2^{k-i-1} x_{i} + 2^{k-i-1} x_{i} = 2^{k-i} x_{i}$ new vertices. 
    In total, we construct a path from $y$ to $x$ through $v$ of link distance at most $3 \cdot 2^{k-i-1}x_i < a_{(k)}$. 
    
    Now we end up with a polygon where the length of each chain is not greater than the requirement. 
    To exactly realize the requirement, one can always subdivide a segment to include more vertices. 
\end{proof}

It's not hard to see that, when $k$ is small enough, the condition in \cref{lem:bound} always holds.

\begin{theorem}
    \label{thm:lower-bound}
    Consider a realizable partial description $P$ of a polygon with $n$ vertices, of which $k$ vertices are specified. 
    If $k \cdot 3^{k+1} \le n$, then for any point $v$, there exists a realization $Q$ such that $v$ is on the boundary of $Q$.
\end{theorem}

\begin{proof}
    Let $i^* = \max\{i \in \{0, 1, \dots, k - 1\}: x_{i} \le k\cdot3^i\}$. Since $x_{0} = k$, $i^*$ is well defined. We will show that $i^*$ satisfies the condition in \cref{lem:bound}.
    For every $j \in \{i^*+1, \dots, k-1\}$, we have $x_j>k\cdot 3^j \ge 3^{j-i^*}x_{i^*}$ by the definition of $i^*$. Then
    \[
        a_{(j)} \ge \frac{x_j-x_{i^*}}{j-i^*} > \frac{3^{j-i^*}-1}{j-i^*} x_{i^*} > 2^{j-i^*-1}x_{i^*},
    \]
    where the first inequality comes from the fact that $a_{(j)}$'s are non-decreasing.
    Also note that $x_k=n \ge k\cdot 3^{k+1} \ge 3^{k-i^*+1}x_{i^*}$, so
    \[
        a_{(k)} \ge \frac{x_k-x_{i^*}}{k-i^*} > \frac{3^{k-i^*+1}-1}{k-i^*} x_{i^*} > 3\cdot 2^{k-i^*-1}x_{i^*}.
    \]
    Therefore, by \cref{lem:bound}, the realization $Q$ exists.
\end{proof}

This immediately implies:

\begin{corollary}
    For any partial description $P$ of a polygon with $n$ vertices, of which at most $\log_3 n - \log_3\log_3 n - 1$ vertices are specified, it is either unrealizable or unbounded. 
\end{corollary}

\printbibliography

@misc{slides,
  author       = "Marc van Kreveld",
  title        = "Lecture 1: Introduction and Convex Hulls",
  howpublished = "",
  month        = "",
  year         = "",
  note         = "Slides currently available at \url{https://ics-websites.science.uu.nl/docs/vakken/ga/2024/}"
}

@article{DBLP:journals/jocg/KostitsynaLPS17,
  author       = {Irina Kostitsyna and
                  Maarten L{\"{o}}ffler and
                  Valentin Polishchuk and
                  Frank Staals},
  title        = {On the complexity of minimum-link path problems},
  journal      = {J. Comput. Geom.},
  volume       = {8},
  number       = {2},
  pages        = {80--108},
  year         = {2017},
  url          = {https://doi.org/10.20382/jocg.v8i2a5},
  doi          = {10.20382/JOCG.V8I2A5},
  bibsource    = {dblp computer science bibliography, https://dblp.org}
}

@incollection{DBLP:reference/cg/Mitchell04,
  author       = {Joseph S. B. Mitchell},
  editor       = {Jacob E. Goodman and Joseph O'Rourke and Csaba D. Tóth},
  title        = {Shortest Paths and Networks},
  publisher = {CRC Press LLC},
  booktitle    = {Handbook of Discrete and Computational Geometry, Third Edition},
  pages        = {811--842},
  year         = {2017},
  url          = {https://www.csun.edu/~ctoth/Handbook/chap31.pdf},
  isbn          = {9781498711395}
}

@article{DBLP:journals/comgeo/MitchellPS14,
  author       = {Joseph S. B. Mitchell and
                  Valentin Polishchuk and
                  Mikko Sysikaski},
  title        = {Minimum-link paths revisited},
  journal      = {Comput. Geom.},
  volume       = {47},
  number       = {6},
  pages        = {651--667},
  year         = {2014},
  url          = {https://doi.org/10.1016/j.comgeo.2013.12.005},
  doi          = {10.1016/J.COMGEO.2013.12.005},
  bibsource    = {dblp computer science bibliography, https://dblp.org}
}

@inproceedings{DBLP:conf/soda/KunnemannN22,
  author       = {Marvin K{\"{u}}nnemann and
                  Andr{\'{e}} Nusser},
  editor       = {Joseph (Seffi) Naor and
                  Niv Buchbinder},
  title        = {Polygon Placement Revisited: (Degree of Freedom + 1)-SUM Hardness
                  and an Improvement via Offline Dynamic Rectangle Union},
  booktitle    = {Proceedings of the 2022 {ACM-SIAM} Symposium on Discrete Algorithms,
                  {SODA} 2022, Virtual Conference / Alexandria, VA, USA, January 9 -
                  12, 2022},
  pages        = {3181--3201},
  publisher    = {{SIAM}},
  year         = {2022},
  url          = {https://doi.org/10.1137/1.9781611977073.124},
  doi          = {10.1137/1.9781611977073.124},
  bibsource    = {dblp computer science bibliography, https://dblp.org}
}

@article{DBLP:journals/theoretics/AbrahamsenMS24,
  author       = {Mikkel Abrahamsen and
                  Tillmann Miltzow and
                  Nadja Seiferth},
  title        = {Framework for {\(\exists\)}{\(\mathbb{R}\)}-Completeness of Two-Dimensional
                  Packing Problems},
  journal      = {TheoretiCS},
  volume       = {3},
  year         = {2024},
  url          = {https://doi.org/10.46298/theoretics.24.11},
  doi          = {10.46298/THEORETICS.24.11},
  bibsource    = {dblp computer science bibliography, https://dblp.org}
}

@inproceedings{DBLP:conf/focs/AbrahamsenS24,
  author       = {Mikkel Abrahamsen and
                  Jack Stade},
  title        = {Hardness of Packing, Covering and Partitioning Simple Polygons with
                  Unit Squares},
  booktitle    = {65th {IEEE} Annual Symposium on Foundations of Computer Science, {FOCS}
                  2024, Chicago, IL, USA, October 27-30, 2024},
  pages        = {1355--1371},
  publisher    = {{IEEE}},
  year         = {2024},
  url          = {https://doi.org/10.1109/FOCS61266.2024.00087},
  doi          = {10.1109/FOCS61266.2024.00087},
  bibsource    = {dblp computer science bibliography, https://dblp.org}
}

@inproceedings{DBLP:conf/soda/AbrahamsenR25,
  author       = {Mikkel Abrahamsen and
                  Nichlas Langhoff Rasmussen},
  editor       = {Yossi Azar and
                  Debmalya Panigrahi},
  title        = {Partitioning a Polygon Into Small Pieces},
  booktitle    = {Proceedings of the 2025 Annual {ACM-SIAM} Symposium on Discrete Algorithms,
                  {SODA} 2025, New Orleans, LA, USA, January 12-15, 2025},
  pages        = {3562--3589},
  publisher    = {{SIAM}},
  year         = {2025},
  url          = {https://doi.org/10.1137/1.9781611978322.118},
  doi          = {10.1137/1.9781611978322.118},
  bibsource    = {dblp computer science bibliography, https://dblp.org}
}

@inproceedings{DBLP:conf/esa/ArkinD0GMPT20,
  author       = {Esther M. Arkin and
                  Rathish Das and
                  Jie Gao and
                  Mayank Goswami and
                  Joseph S. B. Mitchell and
                  Valentin Polishchuk and
                  Csaba D. T{\'{o}}th},
  editor       = {Fabrizio Grandoni and
                  Grzegorz Herman and
                  Peter Sanders},
  title        = {Cutting Polygons into Small Pieces with Chords: Laser-Based Localization},
  booktitle    = {28th Annual European Symposium on Algorithms, {ESA} 2020, September
                  7-9, 2020, Pisa, Italy (Virtual Conference)},
  series       = {LIPIcs},
  volume       = {173},
  pages        = {7:1--7:23},
  publisher    = {Schloss Dagstuhl - Leibniz-Zentrum f{\"{u}}r Informatik},
  year         = {2020},
  url          = {https://doi.org/10.4230/LIPIcs.ESA.2020.7},
  doi          = {10.4230/LIPICS.ESA.2020.7},
  bibsource    = {dblp computer science bibliography, https://dblp.org}
}

@inproceedings{DBLP:conf/approx/FeketeKKMS11,
  author       = {S{\'{a}}ndor P. Fekete and
                  Tom Kamphans and
                  Alexander Kr{\"{o}}ller and
                  Joseph S. B. Mitchell and
                  Christiane Schmidt},
  editor       = {Leslie Ann Goldberg and
                  Klaus Jansen and
                  R. Ravi and
                  Jos{\'{e}} D. P. Rolim},
  title        = {Exploring and Triangulating a Region by a Swarm of Robots},
  booktitle    = {Approximation, Randomization, and Combinatorial Optimization. Algorithms
                  and Techniques - 14th International Workshop, {APPROX} 2011, and 15th
                  International Workshop, {RANDOM} 2011, Princeton, NJ, USA, August
                  17-19, 2011. Proceedings},
  series       = {Lecture Notes in Computer Science},
  volume       = {6845},
  pages        = {206--217},
  publisher    = {Springer},
  year         = {2011},
  url          = {https://doi.org/10.1007/978-3-642-22935-0\_18},
  doi          = {10.1007/978-3-642-22935-0\_18},
  bibsource    = {dblp computer science bibliography, https://dblp.org}
}

\end{document}